\documentclass[journal]{IEEEtran}
\usepackage{amsthm}
\usepackage{graphicx}
\usepackage{subcaption}

\usepackage{cite}
\usepackage{amsmath}
\usepackage{amssymb}
\usepackage{mathtools}

\DeclareMathOperator*{\argmax}{arg\,max}
\usepackage{algorithm}
\usepackage{algorithmic}
\makeatletter
\def\BState{\State\hskip-\ALG@thistlm}
\makeatother
\usepackage{fixltx2e}

\usepackage{stfloats}

\usepackage{xcolor}
\usepackage{lipsum}
\usepackage[margin=0.75in]{geometry}
\newgeometry{top=1.1in, left=0.85in,right=0.85in,bottom=1.07in}

\begin{document}
\setlength{\belowcaptionskip}{-10pt}
\setlength{\abovedisplayskip}{2pt}
\setlength{\belowdisplayskip}{2pt}
\setlength{\parskip}{0pt}
%
\title{An Online Admission Control Mechanism for Electric Vehicles at Public Parking Infrastructures}

\author{\IEEEauthorblockN{Nathaniel Tucker \quad}
\and
\IEEEauthorblockN{Mahnoosh Alizadeh}
\vspace*{-0.65cm}
}


%


\maketitle

\begin{abstract}
We study an online reservation system that allows electric vehicles (EVs) to park and charge at parking facilities equipped with electric vehicle supply equipment (EVSEs). We consider the case where EVs arrive in an online fashion and the facility coordinator must immediately make an admission or rejection decision as well as assign a specific irrevocable parking spot to each admitted EV. By means of strategic user admittance and smart charging, the objective of the facility coordinator is to maximize total user utility minus the operational costs of the facilities. We discuss an online pricing mechanism based on primal-dual methods for combinatorial auctions that functions as both an admission controller and a distributor of the facilities' limited charging resources. We analyze the online pricing mechanism's performance compared to the optimal offline solution and provide numerical results that validate the mechanism's performance for various test cases. 
\end{abstract}


\section*{Notation}
\begin{align}
    &\mathcal{N} &&\text{Set of arriving EVs indexed by }n \nonumber \\[-2.0pt]
    &\mathcal{L} &&\text{Set of charging facilities indexed by }\ell \nonumber \\[-2.0pt]
    &\mathcal{M}_l &&\text{Set of EVSEs at facility $\ell$ indexed by $m$} \nonumber \\[-2.0pt]
    &\mathcal{T} &&\text{Set of time intervals indexed by }t=1,\dots,T \nonumber \\[-2.0pt]
    &\Theta && \text{Set of all possible user types} \nonumber \\[-2.0pt]
    &\mathcal{O}_n && \text{Set of schedule options that satisfy user }n \nonumber \\[-2.0pt]
    &M_l  &&\text{Number of EVSEs at facility }\ell \nonumber \\[-2.0pt]
    &C_l  &&\text{Number of cables per EVSE at facility }\ell \nonumber \\[-2.0pt]
    &E_l  &&\text{EVSE max energy output at facility }\ell \nonumber \\[-2.0pt]
    &s_l(t)  &&\text{Available solar at facility }\ell\text{ at time }t \nonumber \\[-2.0pt]
    &S_l  &&\text{Max solar generation at facility }\ell \nonumber \\[-2.0pt]
    &\pi_l(t)  &&\text{Grid energy price per unit at facility }\ell\text{ at time }t \nonumber \\[-2.0pt]
    &G_l(t)  &&\text{Max energy from grid at facility }\ell\text{ at time }t \nonumber \\[-2.0pt]
    &\theta_n  &&\text{Arrival }n\text{'s user type} \nonumber \\[-2.0pt]
    &t_n^-  &&\text{Arrival }n\text{'s reservation start time} \nonumber \\[-2.0pt]
    &t_n^+  &&\text{Arrival }n\text{'s reservation end time} \nonumber \\[-2.0pt]
    &h_n  &&\text{Arrival }n\text{'s energy request} \nonumber \\[-2.0pt]
    &\{\ell_n\}  &&\text{User }n\text{'s preferred facilities} \nonumber \\[-2.0pt]
    &\{v_{n\ell}\}  &&\text{User }n\text{'s valuations for each facility }\ell \nonumber\\[-2.0pt]
    &c_{no}^{ml}(t) && \text{Cable reservation for user $n$ in option $o$  at}\nonumber \\*[-2.0pt]
    & && \text{EVSE $m$ at facility $\ell$ at time $t$} \nonumber \\[-2.0pt]
    &e_{no}^{ml}(t) && \text{Charge reservation for user $n$ in option $o$ at }\nonumber \\*[-2.0pt]
    & && \text{EVSE $m$ at facility $\ell$ at time $t$} \nonumber \\[-2.0pt]
    &x_{no}^{ml} && \text{Binary assignment variable for user $n$ for}\nonumber \\*[-2.0pt]
    & && \text{option $o$ at EVSE $m$ at facility $\ell$} \nonumber \\[-2.0pt]
    &\hat{p}_{no}^{ml} && \text{Payment from user $n$ for option $o$ at EVSE $m$}\nonumber \\*[-2.0pt]
    & && \text{at facility $\ell$} \nonumber \\[-2.0pt]
    &y_{c}^{ml}(t) && \text{Cables allocated at EVSE $m$ at facility $\ell$ at time $t$}\nonumber \\*[-2.0pt]
    &y_{e}^{ml}(t) && \text{Energy allocated at EVSE $m$ at facility $\ell$ at time $t$}\nonumber \\*[-2.0pt]
    &y_{g}^{l}(t) && \text{Total energy needed at facility $\ell$ at time $t$}\nonumber \\*[-2.0pt]
    &f_g^l(\cdot)  &&\text{Facility }\ell\text{'s electricity procurement cost function} \nonumber \\[-2.0pt]
    &u_n &&\text{User $n$'s utility from the EVSE reservation system} \nonumber \\[-2.0pt]
    &p_c^{ml}(t) &&\text{Cable price at EVSE $m$ at facility $\ell$ at time $t$} \nonumber \\[-2.0pt]
    &p_e^{ml}(t) &&\text{Charging price at EVSE $m$ at facility $\ell$ at time $t$} \nonumber \\[-2.0pt]
    &p_g^{l}(t) &&\text{Energy procurement price at facility $\ell$ at time $t$} \nonumber \\[-2.0pt]
    &f^*(\cdot)  &&\text{Fenchel conjugate of a cost function/constraint} \nonumber \\[-2.0pt]
    &L_{c,e,g}  &&\text{Lower bound on valuations per resource} \nonumber \\[-2.0pt]
    &U_{c,e,g}  &&\text{Upper bound on valuations per resource} \nonumber \\[-2.0pt]
    &R  &&\text{Number of resources available across all facilities} \nonumber \\[-2.0pt]
    &\alpha &&\text{Online mechanism's competitive ratio} \nonumber \\[-2.0pt]
    &\underbar{s}_l(\cdot) &&\text{Lower bound on available solar energy}\nonumber \\[-2.0pt]
    &\overline{s}_l(\cdot) &&\text{Upper bound on available solar energy}\nonumber \\[-2.0pt]
    &I_C &&\text{EVSE investment cost per cable}\nonumber \\[-2.0pt]
    &I_M &&\text{EVSE installation cost}\nonumber \\[-2.0pt]
    &I_{m,n} &&\text{Infrastructure maintenance and networking cost}\nonumber
\end{align}

%
\IEEEpeerreviewmaketitle
\newtheorem{proposition}{Proposition}
\newtheorem{corollary}{Corollary}[proposition]
\makeatletter
\def\blfootnote{\xdef\@thefnmark{}\@footnotetext}
\makeatother

\blfootnote{
\indent
This work was supported by the California Energy Commission through SLAC. Solicitation: GFO-16-303. Agreement: EPC-16-057.\\
\hspace{8.5pt}N. Tucker and M. Alizadeh are with the Department of Electrical and Computer Engineering, University of California, Santa Barbara, CA 93106 USA (emails: nathaniel\_tucker@ucsb.edu; alizadeh@ucsb.edu).}

\section{Introduction}

As of October 2018, one million plug-in electric vehicles (PEVs) have been sold in the United States \cite{New_stat}. Furthermore, sales have exceeded 20,000 units per month since May 2018 and these numbers are expected to continue trending upward beyond 2020 \cite{New_stat}. As such, coordinated charging strategies and charging infrastructure planning are paramount for ensuring the growing charging demand is satisfied in an environmentally responsible manner.

There has been a growing number of related papers that study EV smart charging methods as well as infrastructure planning and investment analysis to encourage renewable energy usage in vehicle fleets, aggregate groupings, and parking facilities. For an overview, \cite{New_TSG_SmartEVGrid,New_Survey,New_20} provide in-depth reviews of smart charging technologies as well as societal and grid impacts. Investigations on the interactions between EV aggregations and the grid can be found in \cite{New_24,New_22}. 
Because smart charging has proven to benefit society, infrastructure investments must be made to support future charging implementations \cite{New_17,New_15,New_13,New_14,New_12,New_11}. Papers \cite{New_11,New_12,New_13} study where to locate charging stations as well as how to effectively size the facilities. In \cite{New_17}, the authors study a planning framework for charging stations from the perspective of a social planner. Likewise, the authors in \cite{New_15} study the design criteria for Fast Charging Stations (FCSs) based on mobility behaviors and paper \cite{New_14} studies a planning scheme to maximize FCS usage and minimize infrastructure costs.

A critical but less studied problem is that coordinated charging at infrastructures can be heavily stunted if usage of the EVSEs is left uncontrolled \cite{bryce}. Without EV routing within parking facilities, EVSEs at preferred locations (e.g., near an elevator) can become congested while other EVSEs are left empty. This limits the smart charging benefits as congested EVSEs are forced to charge one EV after another to satisfy charging demand. Similarly, without admission control the limited charging resources at facilities could be allocated to low priority users, (e.g., users with small charging demands, users with long sojourn times, or users who are willing to park elsewhere) therefore, occluding high priority users that arrive later in the day. As such, the focus of this paper is to jointly perform admission control and smart charging, complementing previous work on coordinated charging and infrastructure planning. 

Prior work in this area includes \cite{New_19} where the authors investigate both First-Come-First-Serve (FCFS) and State-of-Charge (SoC) threshold policies for discerning which EVs are granted permission to use the EV charging infrastructure. Paper \cite{Robu} studies an online mechanism for the allocation of electricity to a population of EVs that have non-increasing marginal value for energy. Their setting allows for cancellation of reservations, which in our case is not allowed. In \cite{Zheng}, an online algorithm for scheduling deferrable charging requests to balance the total value of vehicle owners and the total cost for providing charging service is studied, but they also allow for revocation of previously allocated resources. Paper \cite{auc2charge} investigates an online auction that allows EV users to submit bids on their charging demand to the charging station and then the mechanism makes corresponding electricity allocation and pricing decisions. In this approach, users are expected to update their bids while parked at the charging station instead of only submitting one initial bid. Additionally, \cite{Stein} studied a consensus approach for an online setting where selfish EVs compete for a limited amount of energy. An intelligent parking lot energy management system is studied in \cite{Honarmand2014OptimalSO} to manage the scheduling of EVs to maximize charging for all EVs. Moreover, paper \cite{differentiated} formulates and analyzes a market model for deadline-differentiated pricing of deferrable electric power services; however, it does not focus on high levels of congestion or adversarial user valuations. Paper \cite{menu} presents a menu-based pricing scheme for allocating charge time within a facility and can lead to an efficient alternative approach to the EVSE reservation problem. The mechanism we present in this manuscript focuses more on the congestion within facilities due to limited number of EVSEs and high demand with the objective of admitting highest priority users.
In most previous work, charging facilities are assumed to have traditional Single-Output-Single-Cable (SOSC) EVSEs. Recently, a more versatile charger has been gaining popularity: the Single-Output-Multiple-Cable (SOMC) EVSE which allows multiple EVs to be connected to the same charger, but only one EV receives charge at a time \cite{SOMC}. SOMC EVSEs can improve facility operations by allowing more flexibility in charge scheduling and decreasing idle plug-in time from traditional SOSC chargers. Furthermore, SOMC chargers eliminate the need for users to remove their vehicles once their charging session is complete. In any SOMC facility charging strategy, these idle EVs need to be accounted for; if not, the revenue of the facility will be reduced (our solution accounts for the times when EVs are charging \textit{and} when they are idle). Utilizing SOMC EVSES, the authors of \cite{SOMC} study infrastructure investments, the authors of \cite{bryce} study centralized online assignment methods such as Next-Fit and Worst-Fit for arriving EVs at a parking facility, and the authors of \cite{ntucker_allerton} study multiple online pricing heuristics for EV to EVSE allocation to increase smart charging capabilities.

In this manuscript, we present an online pricing mechanism that functions as both an admission controller for parking facility access and a resource manager that optimizes smart charging strategies for vehicles admitted to the facility. The work presented in this manuscript complements existing literature in the area and the main contributions are as follows: 
\begin{enumerate}
    \item The online mechanism readily accommodates multiple locations, multiple limited resources, operational costs, and renewable generation integration. 
    \item The online mechanism does not rely on fractional allocations or rounding methods to produce integer allocations in a computationally feasible manner and it never revokes previously made reservations.
    \item The online mechanism readily handles the inherent stochasticity of the EVSE reservation problem including unknown sojourn times, unknown energy requests, and unknown user valuation functions. 
    \item The online mechanism is robust to adversarially chosen arrival sequences and always yields social welfare within a factor of $\frac{1}{\alpha}$ of the offline optimal. 
\end{enumerate}
\noindent Preliminary results from this paper were previously submitted as a conference paper\cite{ntucker_ACC}. In this paper, we present new theoretical results on the pricing functions by accounting for time varying behind-the-meter solar generation and provide more extensive numerical results showing the performance of the mechanism.

The remainder of the paper is organized as follows. Section \ref{section: prob_form} presents the system model and describes the offline EVSE reservation problem. Section \ref{section: Online_Mech} presents the online auction mechanism used to provide an approximate solution to the EVSE reservation problem and discusses the online mechanism's properties and performance guarantees. Section \ref{section:numerical} presents simulation results that showcase the performance of the online mechanism for different test cases.

\section{System Model}
\label{section: prob_form}
\subsection{System Structure and User Characteristics}

In this section, we describe our model for the EVSE reservation problem and user characteristics. We consider an EVSE reservation system that controls $L$ dispersed parking facilities. Each parking facility $l\in L$ is equipped with $M_l$ SOMC EVSEs each with $C_l$ cables (i.e., each facility can park $M_lC_l$ EVs at any given time $t=1,\dots,T$ but only charge $M_l$). In addition to the cable constraints, each EVSE has a maximum power output constraint denoted by $E_l$ that limits the amount of energy the EVSE can deliver in one time slot. To supply the EVSEs with electricity, each parking facility can procure energy from two sources: a rooftop solar generation system or the local distribution grid. We denote the available solar energy at facility $l$ at time $t$ with the variable $s_l(t)\in [0,S_l]$ where $S_l$ is the maximum rating for facility $l$'s rooftop system. Additionally, we denote $\pi_l(t)$ as the per unit price of electricity from the grid. Due to physical limits of the local transformer, we constrain facility $l$ to procure no more than $G_l(t)$ units of energy from the grid at each time slot.

Each day, $N$ EV owners submit requests to park and charge at various facilities. Each EV owner (user) is characterized by a set of attributes. Suppose user $n$ wants to park and charge her EV. When she submits her reservation request at time $t_n$, she commits to arrive at one of her desired parking facilities $\{l_n\}$ at time $t_n^-$ and to depart at $t_n^+$. Furthermore, user $n$ receives value $\{v_{nl}\}$ if her EV receives $h_n$ units of energy from facility $l$, meaning users have preferences for different facilities. With the aforementioned nomenclature, each arrival can be characterized by user `type':
\begin{equation}
\theta_n = \{ t_n^-, t_n^+, h_n, \{l_n\}, \{v_{nl}\} \} \in \Theta,
\end{equation}
where $\Theta$ is the type space of all possible users. Fig. \ref{fig: Model} presents an example allocation sequence with 4 arrivals and 2 EVSEs. Specifically, Fig. \ref{fig: Model} showcases the fact that there are limited charging resources within the parking facility that need to be allocated to the arrivals. Each arriving EV needs an EVSE cable, an EVSE energy schedule, and the facility needs an energy procurement schedule from i) behind-the-meter solar, ii) the local distribution grid, or iii) a combination of solar and grid energy. The arriving vehicles enter the facility one-by-one and utilize the limited resources during their stay, affecting how future arrivals are allocated as seen in Fig. \ref{fig: Model}.

\subsection{Offline Problem Formulation}
\begin{figure}[]
    \centering
    \includegraphics[width=\columnwidth]{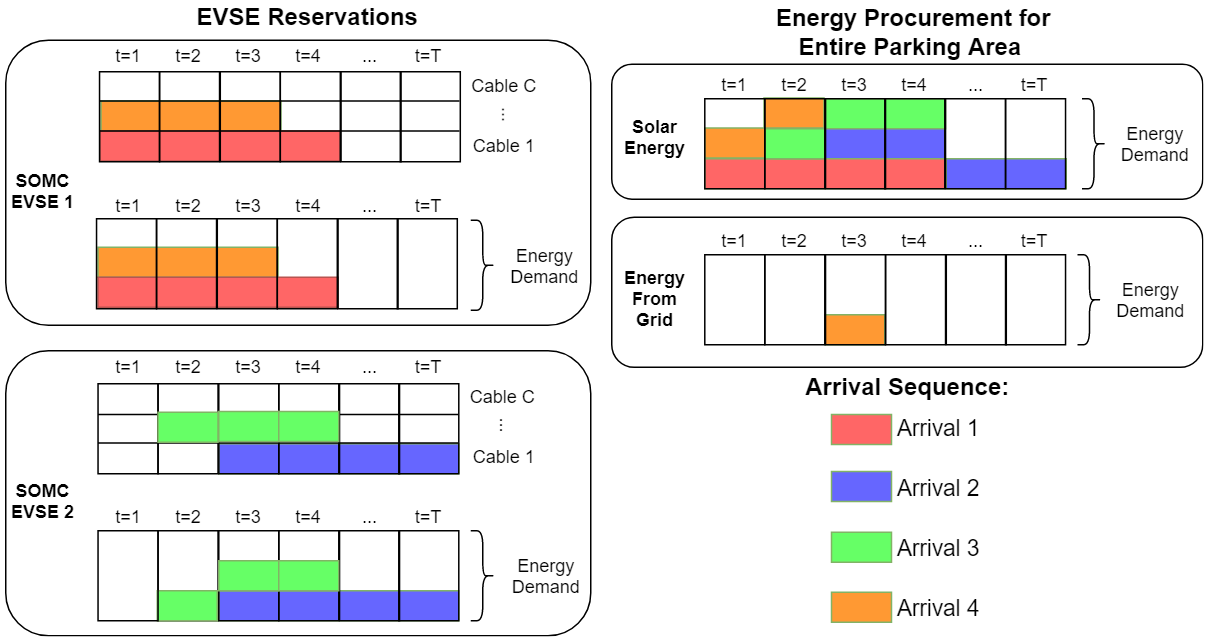}
    \caption{Example reservation schedule.} 
    \label{fig: Model}
\end{figure}
To request a reservation, user $n$ submits her user type $\theta_n$ to the EVSE reservation system. The EVSE reservation system creates a set of possible schedules that will fulfill user $n$'s requirements. Namely, each possible schedule, or \textit{option}, contains a cable reservation for the entire parking duration and the charge schedules that sum up to her desired charge amount. The reservation system generates these options for each facility within user $n$'s desired facility set and then the option that yields the highest utility to the user is selected. 

We denote the set of options (potential schedules) as $\mathcal{O}_n$. Each option $o\in\mathcal{O}_n$ corresponds to a facility $l_n$, a cable reservation $c_{no}^{ml}(t)$, and a charge schedule $e_{no}^{ml}(t)$. The cable reservation $c_{no}^{ml}(t)$ takes values 0 or 1 depending if user $n$ is assigned a cable from EVSE $m$ at facility $l$ at time $t$ in option $o$. Similarly, $e_{no}^{ml}(t)$ takes values from a discrete set corresponding to the energy delivered to user $n$'s EV. Through $e_{no}^{ml}(t)$, the EVSE reservation system is able to customize when each EV will receive charge and when it will be idle as well as the rate of charge. The set of feasible options for user $n$ can be written as: 
\begin{equation}
    \{ t_n^-, t_n^+, \{ c_{no}^{ml}(t) \}, \{ e_{no}^{ml}(t) \}, \{l_n\}, \{v_{nl}\} \}.
\end{equation}

When deciding whether or not to admit user $n$ and which option to allocate, the reservation system sets the binary variable $x_{no}^{ml}$ to 1 if option $o$ is chosen at EVSE $m$ at facility $l$. Additionally, the reservation system computes payments $\hat{p}_{no}^{ml}$ for each option that the user pays if accepted. If a user is not admitted into any parking facilities, she receives zero value and parks in an auxiliary lot without EVSEs.
    
The EVSE reservation system keeps track of the allocated resources throughout the day. The variables $y_c^{ml}(t)$ and $y_e^{ml}(t)$ correspond to the allocated cables and energy, respectively, at EVSE $m$, facility $l$, at time $t$. Each facility also has to procure the energy needed by all the EVSEs within; therefore, the total energy needed by facility $l$ at time $t$ is denoted as $y_g^{l}(t)$. Equations \eqref{eqn:offline cable demand}-\eqref{eqn:offline gen demand} detail how each resource demand is calculated:
\begin{align}
    \label{eqn:offline cable demand}
    & y_{c}^{ml}(t)=\sum_{\mathcal{N},\mathcal{O}_n} c_{no}^{ml}(t) x_{no}^{ml},\\
    \label{eqn:offline energy demand}
    & y_{e}^{ml}(t)=\sum_{\mathcal{N},\mathcal{O}_n} e_{no}^{ml}(t) x_{no}^{ml},\\
    \label{eqn:offline gen demand}
    & y_{g}^{l}(t)=\sum_{\mathcal{N},\mathcal{O}_n,\mathcal{M}_l} e_{no}^{ml}(t) x_{no}^{ml}.
\end{align}

The energy procurement, $y_g^{l}(t)$, determines the operational cost of facility $l$:
\begin{align}
\label{eq:generation cost}
    &f_{g}^{l}(y_{g}^{l}(t)) = \\
    &\nonumber\begin{cases}
        0 & y_g^{l}(t) \in [0,s_l(t)\big)\\
        \pi_l(t) (y_{g}^{l}(t)-s_l(t)) & y_g^{l}(t) \in [s_l(t),s_l(t)+G_l(t)] \\
        +\infty & y_{g}^{l}(t) > s_l(t)+G_l(t).
        \end{cases}
\end{align}

The operational cost of the facility is zero while solar energy is available. Once the demand, $y_g^{l}(t)$, exceeds the available solar, energy is purchased from the grid. Once the demand exceeds the sum of available behind-the-meter solar energy and the  transformer limit, no more energy can be procured.

With the system variables and equations defined, we can write the offline social welfare maximization problem (assuming all users' information is known beforehand):
\begin{subequations}
\begin{align}
    \label{eqn:offline obj}
    &\max_{x} \sum_{\mathcal{N},\mathcal{O}_n,\mathcal{L},\mathcal{M}_l} v_{nl} x_{no}^{ml}-\sum_{\mathcal{T}, \mathcal{L}}    f_{g}^{l}(y_{g}^{l}(t)) \\ \nonumber
    &\textrm{ subject to:}\nonumber\\
    \label{eqn:offline one option}
    &\sum_{\mathcal{O}_n, \mathcal{L},\mathcal{M}_{l}} x_{no}^{ml} \leq 1, \quad\forall\; n\\
    &\label{eqn:offline integer}\hspace{1pt} x_{no}^{ml}  \in \{ 0,1 \}, \hspace{11pt}\forall \; n,  o, l,m\\
    &\label{eqn:offline cable lim}\hspace{1pt}y_{c}^{ml}(t) \leq C_l, \hspace{13pt} \forall\; l,m,t\\
    &\label{eqn:offline energy lim}\hspace{1pt}y_{e}^{ml}(t) \leq E_l, \hspace{13pt} \forall\; l,m,t\\
    & \nonumber \textrm{ and }\eqref{eqn:offline cable demand}, \eqref{eqn:offline energy demand}, \eqref{eqn:offline gen demand}.
\end{align}
\end{subequations}
Moreover, the objective \eqref{eqn:offline obj} is to maximize the total \textit{social welfare} of the system. This includes the utility gained by arrivals using the system minus the operational costs of the facilities (we note that users who are not admitted receive utility equal to zero). Constraints \eqref{eqn:offline one option}-\eqref{eqn:offline energy lim} respectively ensure at most one option is selected per user, the assignment variable is an integer, the cable demand does not exceed capacity, and the energy demand does not exceed capacity. Equations \eqref{eqn:offline cable demand}-\eqref{eqn:offline gen demand} sum the resource demands.

Temporarily relaxing the integrality constraint \eqref{eqn:offline integer} on $x_{no}^{ml}$  allows us to find the Fenchel dual problem with dual variables $u_n$, $p_c^{ml}(t)$, $p_e^{ml}(t)$, and $p_g^{l}(t)$ \cite{Fenchel}. In the following, the Fenchel conjugate of a function is given as:
\begin{equation}
    f^*(p(t)) = \sup_{y(t)\geq0} \big\{ p(t)y(t) - f(y(t)) \big\}.
\end{equation} 
Accordingly, the Fenchel dual of \eqref{eqn:offline obj}-\eqref{eqn:offline energy lim} can be written:

\begin{subequations}
\begin{alignat}{3}
    \label{eqn:dual obj}
    &\min_{u,p} \sum_{\mathcal{N}} u_n + \sum_{\mathcal{T}, \mathcal{L}}    f_{g}^{l*}(p_{g}^{l}(t))\\* 
    &\nonumber+ \sum_{\mathcal{T}, \mathcal{L},\mathcal{M}_l} \Big( f_{c}^{ml*}(p_{c}^{ml}(t)) + f_{e}^{ml*}(p_{e}^{ml}(t)) \Big)
     \\
    &\nonumber \textrm{ subject to:}\\*
    \label{eqn:dual user utility}
    &\;u_n \geq \;v_{nl} 
    - \sum_{\mathcal{T}} \Big( c_{no}^{ml}(t)p_{c}^{ml}(t)\\*
    &\hspace{42pt}+e_{no}^{ml}(t)\big(p_{e}^{ml}(t)+p_{g}^{l}(t)\big)\Big)\hspace{10pt}\forall\;n,o,l,m\nonumber\\*
    \label{eqn:dual ui}
    &u_n \geq \;0, \hspace{90pt}\forall \;n\\*
    \label{eqn:dual pj}
    &p_{c}^{ml}(t),\;p_{e}^{ml}(t),\;p_{g}^{l}(t) \geq \;0,  \hspace{12pt}\forall\; l,m,t,
\end{alignat}
\end{subequations}

\noindent where $f^*(p(t))$ is the Fenchel conjugate for the limited resources' dual variables. The Fenchel conjugates for the capacity constraints can be written as:

\begin{align}
\label{eq:fenchel cable cost}
    &f^{ml*}_{c}(p^{ml}_{c}(t)) = p^{ml}_{c}(t)C_l, \hspace{24pt} p^{ml}_{c}(t) \geq 0\\
    \label{eq:fenchel energy cost}
    &f^{ml*}_{e}(p^{ml}_{e}(t)) = p^{ml}_{e}(t)E_l, \hspace{24pt} p^{ml}_{e}(t) \geq 0.
\end{align}
Additionally, the Fenchel conjugate for the energy procurement operational cost function can be written as:
\begin{align}
\label{eq:fenchel gen cost}
    &f^{l*}_{g}(p^{l}_{g}(t)) =  \\*
    &\nonumber\begin{cases}
        s_l(t)p^{l}_{g}(t), &p^{l}_{g}(t) < \pi_l(t) \\
        (s_l(t)+G_l(t))p^{l}_g(t)-G_l(t)\pi_l(t) & p^{l}_{g}(t) \geq \pi_l(t). \\
    \end{cases}
\end{align}
\subsection{Admittance, Rejection, and Allocation Decisions}
To determine how the EVSE reservation system decides whether or not to admit a user as well as which option to select if admitted, we make use of the Fenchel dual \eqref{eqn:dual obj}-\eqref{eqn:dual pj}. Specifically, we examine the KKT conditions for constraint \eqref{eqn:dual user utility}. If a user is denied in the offline problem, $u_n$ will be 0; otherwise, if a user is admitted, $u_n$ will be positive. As such, the EVSE reservation system solves the following equation to determine user $n$'s acceptance and her resource allocation:
\begin{align}
\label{eq: u_n}
    &u_n =\max\Big\{0,\max_{\mathcal{O}_n,\mathcal{L},\mathcal{M}_l} \big\{v_{nl} 
    \\*\nonumber&- \sum_{t\in[t_n^-,t_n^+]} \big( c_{no}^{ml}(t)p_{c}^{ml}(t)+e_{no}^{ml}(t)(p_{e}^{ml}(t)+p_{g}^{l}(t) )\big)\big\}\Big\}.
\end{align}
We note that $u_n$ corresponds to user $n$'s utility from the EVSE reservation system. If admitted, the cable reservation and charge schedule chosen for user $n$ correspond to the option $o$, EVSE $m$, and facility $l$ that maximize the second term in equation \eqref{eq: u_n}. Furthermore, the dual variables $p_c^{ml}(t)$, $p_e^{ml}(t)$, and $p_g^{l}(t)$ correspond to prices that users must pay for cables, energy, and energy procurement. As such, the total payments corresponding to user $n$'s different options are calculated as:
\begin{align}
\label{eq: payment}
    \hat{p}_{no}^{ml} = \sum_{\mathcal{T}} \Big( c_{no}^{ml}(t)p_{c}^{ml}(t)+e_{no}^{ml}(t)(p_{e}^{ml}(t)+p_{g}^{l}(t)) \Big).
\end{align}

The EVSE reservation system is allocating options that maximize each user's utility with respect to the current marginal prices. Additionally, users receive non-negative utility for participating in the EVSE reservation system; therefore, we satisfy individual rationality constraints.

We would like to note that our proposed mechanism can also be used without any actual payments if users do not have the option of choosing their type (i.e., their types are preassigned). In a company, if users are assigned valuations (e.g., CEO has a high value and regular employee has lower value, or someone with high charge level has lower value), then the prices do not have to be economic incentives. Rather they are used as dual variables that guide each user’s allocation without any monetary transfer (i.e., each employee does not actually have to pay to use the infrastructure, but different employees have different valuations and the ``shadow prices" allow for quick allocations).

The optimization problems presented in \eqref{eqn:offline obj}-\eqref{eqn:offline energy lim} and \eqref{eqn:dual obj}-\eqref{eqn:dual pj} assume complete knowledge of the arrivals beforehand. In practice, this is not the case; rather, users arrive and depart throughout the day. As such, the solution needs to be an online mechanism that can immediately allocate an arrival without knowledge of the future sequence of arrivals. Additionally, once a user has parked her car within a charging facility, she should not be asked to prematurely move her EV before her departure time. As such, the online mechanism should never revoke previous allocations. In the following, we discuss an online allocation mechanism that solves the EVSE reservation problem and meets the aforementioned design goals.

\section{Online Allocation Mechanism}
\label{section: Online_Mech}

\subsection{Online Marginal Prices}
It is evident that the EVSE reservation problem requires an online solution. In many online problems, approximate dynamic programming (ADP) heuristics have good performance given accurate statistics even with large state-spaces \cite{powell_book, Bertsekas08approximatedynamic,ambulance,ADP_STOC}. However, performance guarantees can be very hard to obtain for multi-stage decision making problems with complex action spaces over long time periods, and in our case, nonstationary arrival patterns and variable forecasts prohibit many traditional ADP techniques. 
As such, we present an online pricing mechanism that calculates the marginal prices on EVSE cables, energy, and generation based on a pricing heuristic, for which we provide performance guarantees. Specifically, our EVSE reservation system updates the prices $p(t)$ heuristically as the amounts of allocated resources $y(t)$ evolve, but only based on past observations. The pricing scheme has two major goals: (1) to make sure that the marginal gain in welfare from an allocation is greater than the operational cost incurred to serve the allocation, and (2) to filter out low value users early to ensure there are adequate resources for higher value users later on. The structure of the marginal price functions we use is similar to that of \cite{IaaS}, where the authors present a pricing framework for cloud-computing systems utilizing data centers with limited computation resources and server costs under an adversarial setting. For the limited number of cables at each EVSE, the proposed marginal payment function is given by:
\begin{align}
\label{eqn:zero inf price}
    p_{c}^{ml}(y_{c}^{ml}(t)) =& \Big(\frac{L_c}{2\sum_{\mathcal{L}}M_l(C_l+E_l+\frac{1}{M_l})}\Big) \\*
    &\nonumber \times\Big( \frac{2\sum_{\mathcal{L}}M_l(C_l+E_l+\frac{1}{M_l})U_c}{L_c} \Big)^{\frac{y_{c}^{ml}(t)}{C_l}},
\end{align}
where $y_{c}^{ml}(t)$ is the current demand for the cables at EVSE $m$ at location $l$ at time $t$. Additionally, $L_c$ and $U_c$ are the lower and upper bounds on users' valuation per cable per unit of time, respectively:
\begin{subequations}
\begin{align}
    & \label{eqn:Lc}L_c = \min_{\mathcal{N},\mathcal{O}_n,\mathcal{L},\mathcal{M}_l} \frac{v_{nl}}{\sum_{\mathcal{L}}M_l(C_l+E_l+\frac{1}{M_l})\sum_{t\in[t_n^-,t_n^+]}c_{no}^{ml}(t)},\\
    &\label{eqn:Uc}U_c = \max_{\mathcal{N},\mathcal{O}_n,\mathcal{L},\mathcal{M}_l,\mathcal{T}} \frac{v_{nl}}{c_{no}^{ml}(t)},\quad c_{no}^{ml}(t)\neq 0.
\end{align}
\end{subequations}
The pricing function for the EVSE energy units is the same as \eqref{eqn:zero inf price} with the exponent changed to $E_l$ instead of $C_l$. Likewise, calculate $L_e$ and $U_e$ using $e_{no}^{nl}(t)$ in \eqref{eqn:Lc} and \eqref{eqn:Uc}. Additionally, for the energy procurement resource, $L_g$ and $U_g$ are the same as $L_e$ and $U_e$, respectively.

To explain this pricing function, set $y_{c}^{ml}(t)=0$ and \eqref{eqn:zero inf price} outputs a price low enough that any user will be accepted (subject to $L_c$). Moreover, the pricing function \eqref{eqn:zero inf price} yields low initial values to allow reservations early on. As more arrivals are admitted into the reservation system, congestion begins to affect the shared resources. To combat congestion and filter our low value arrivals, as $y_{c}^{ml}(t)$ increases, the prices from \eqref{eqn:zero inf price} increase exponentially. When $y_{c}^{ml}(t)$ is equal to the capacity of the limited resource, the marginal price is set high enough to reject all future arrivals to ensure that no resource will ever be overallocated (we assume $L_c$ and $U_c$ are known).

Designing a pricing function for energy procurement at each facility is more complicated than the cable pricing. Here, the cost to procure energy is piecewise linear and depends on the current solar generation and the transformer capacity. As such, we propose the pricing function:
\begin{align}
\label{eqn:newpricing}
    &p_g^l(y_g^l(t)) = \\*
    &\nonumber\begin{cases}
        \Big( \frac{L_g}{2R} \Big) \Big( \frac{2R\pi_l(t)}{L_g} \Big)^{ \frac{y_g^l(t)}{s_l(t)}}, \hspace{87pt} y_g^l(t) < s_l(t), \\
        \Big( \frac{L_g-\pi_l(t)}{2R} \Big) \Big( \frac{2R(U_g-\pi_l(t))}{L_g-\pi_l(t)} \Big)^{ \frac{y_g^l(t)}{s_l(t)+G_l(t)}}+\pi_l(t), \\
        \hspace{175pt} y_g^l(t) \geq s_l(t), \\
        \end{cases}\\
    &\nonumber\textrm{where } R=\sum_{\mathcal{L}}M_l(C_l+E_l+\frac{1}{M_l}).
\end{align}
Equation \eqref{eqn:newpricing} is similar to the pricing function for the EVSE cables and energy; however, because procuring energy from the grid has non-zero cost, we need to ensure each user's payment is greater than the electricity cost needed to charge their vehicle. Additionally, when a facility's energy demand is less than the available solar, the marginal energy procurement price is reduced below the cost of electricity $\pi_l(t)$ to promote solar consumption. 
\begin{figure}[h]
    \centering
    \includegraphics[width=0.75\columnwidth]{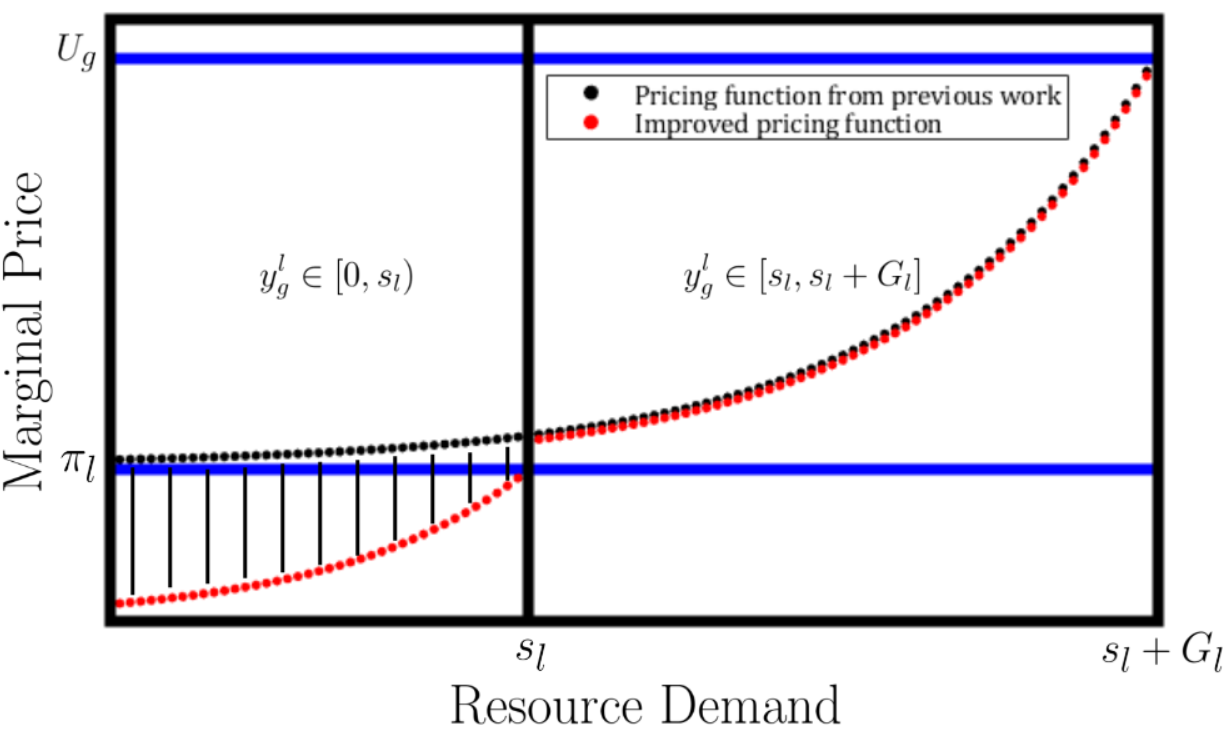}
    \caption{Pricing function for energy procurement. Shaded area: increase in users' utilities from updated pricing function.}
    \label{fig: pricing_example}
\end{figure}

\subsection{Proposed Algorithm and Performance Guarantees}
The admittance, allocation, and price update procedure for the EVSE reservation system is presented in Algorithm \textsc{OnlineParkNCharge}. 
When arrival $n$ submits her request, the system generates the feasible options $\mathcal{O}_n$ that fulfill her demands. Then, the system accepts or rejects user $n$ depending on her potential utility gain due to her valuation and the current resource prices (line \ref{alg:utility}). 
We note that line 9 requires solving an integer constrained maximization problem. This is not computationally burdensome as the optimization is solved for each individual vehicle at the time of arrival, with the potential utilities for each option can be calculated quickly via multiplication and addition. Then, any sorting method can be used to find the highest utility option.
The algorithm updates the primal variables $x_{no}^{ml}$ after each acceptance and rejection. The total resource demands are updated in line \ref{alg:demand update} if user $n$ is accepted into the system. Similarly, the marginal resource prices are updated accordingly in line \ref{alg:price update}.

\begin{algorithm}[]
\small
    \caption{\textsc{OnlineParkNCharge}}
    \label{algorithm}
    \begin{algorithmic}
    \STATE \textbf{Input:} $\mathcal{L}, \mathcal{M}_l,C_l, E_l, G_l, S_l, \pi_l, L_{c,e,g}, U_{c,e,g}$
    \STATE \textbf{Output:} $x, p$
    \end{algorithmic}
    \begin{algorithmic}[1]
    \STATE Define $f_g^l(y_g^l(t))$ according to \eqref{eq:generation cost}.
    \STATE Define the pricing functions $p(y(t))$ according to \eqref{eqn:zero inf price} and \eqref{eqn:newpricing} for cables, energy, and generation.
    \STATE Initialize $x_{no}^{ml}=0$, $y^{ml}(t)=0$, $u_n=0$.
    \STATE Initialize prices $p(0)$ according to \eqref{eqn:zero inf price} and \eqref{eqn:newpricing}.
    \STATE \textbf{Repeat for all $N$ users:}
    \STATE User $n$ submits $\theta_n$, generate feasible charging options.
    \STATE Update dual variable $u_n$ according to \eqref{eq: u_n}. \label{alg:utility}
    \IF{$u_n > 0$}
        \STATE $(o^{\star},m^{\star},l^{\star}) =\argmax_{\mathcal{L}, \mathcal{M}_{l},\mathcal{O}_n}\big\{v_{nl}$ \\
        \vspace{3pt}
        \hspace{20pt} $- \sum_{t\in[t_n^-,t_n^+]} \big( c_{no}^{ml}(t)p_{c}^{ml}(t)$\\
        \vspace{3pt}
        \hspace{20pt} $ +e_{no}^{ml}(t)(p_{e}^{ml}(t)+p_{g}^{l}(t)) \big)\big\}$
        \vspace{3pt}
        \STATE $\hat{p}_{no^{\star}}^{m^{\star}l^{\star}} = \sum_{t\in[t_n^-,t_n^+]} \Big( c_{no^{\star}}^{m^{\star}l^{\star}}(t)p_{c}^{m^{\star}l^{\star}}(t)$\\*
        \hspace{20pt} $+e_{no^{\star}}^{m^{\star}l^{\star}}(t)(p_{e}^{m^{\star}l^{\star}}(t) +p_{g}^{l^{\star}}(t)) \Big)$
        \vspace{3pt}
        \STATE $x_{no^{\star}}^{m^{\star}l^{\star}}=1$ and $x_{no}^{ml}=0$ for all $(o,l,m) \neq (o^{\star},l^{\star},m^{\star})$
        \STATE Update total demand $y(t)$ for cables, energy, and generation according to \eqref{eqn:offline cable demand}-\eqref{eqn:offline gen demand}. \label{alg:demand update}
        \STATE Update marginal prices $p(t)$ for cables, energy, and generation according to \eqref{eqn:zero inf price} and \eqref{eqn:newpricing}. \label{alg:price update}
    \ELSE
        \STATE $x_{no}^{ml}=0$, \hspace{3pt}  $\forall$ $\mathcal{L}$, $\mathcal{M}_l$ and $\mathcal{O}_n$.
    \ENDIF 
    \IF{$\exists o^{\star},m^{\star},l^{\star}$ and $x_{no^{\star}}^{m^{\star}l^{\star}}=1$}
        \STATE Accept user $n$ and allocate cables and energy in parking location $l^{\star}$ at EVSE $m^{\star}$.
        \STATE Charge user $n$ at $\hat{p}_{no^{\star}}^{m^{\star}l^{\star}}$.
    \ELSE
        \STATE Send user $n$ to auxiliary parking.
    \ENDIF
    \end{algorithmic}
\end{algorithm}

Next, we compare the total social welfare resulting from the online solution to the optimal offline solution. Specifically, an online mechanism is said to be $\alpha$-competitive when the ratio of social welfare from the optimal offline solution to the social welfare from the mechanism is bounded by $\alpha\geq1$. We extend a competitive ratio performance guarantee from \cite{IaaS} in Proposition \ref{piecewise linear price}. In the following, to ensure no user purchases too large of a fraction of the total available resource, we assume each user's resource demands are much smaller than the capacity limits.

\begin{proposition}
\label{piecewise linear price}
The marginal pricing function \eqref{eqn:newpricing} is $\alpha_1$-competitive in social welfare when selling limited resources with the piecewise linear operational cost in \eqref{eq:generation cost} where
\begin{align}
\nonumber
    \alpha_1= 2\max_{\mathcal{L},\mathcal{T}}\Big\{\ln{\Big(\frac{2\sum_{\mathcal{L}}M_l(C_l+E_l+\frac{1}{M_l})(U_g-\pi_l(t))}{L_g-\pi_l(t)}\Big)}\Big\}
\end{align}
with the assumption $\sum_{\mathcal{L}}M_l(C_l+E_l+\frac{1}{M_l})\geq\lceil\frac{eL_g}{2\max_{\mathcal{L},\mathcal{T}}\pi_l(t)}\rceil$.
\end{proposition}
\begin{proof}
For $\alpha$-competitiveness, the authors of \cite{IaaS} show that 
marginal pricing functions, operational cost functions, and Fenchel conjugates for the limited resources need to satisfy the \textit{Differential Allocation-Payment Relationship} given by:
\begin{align}
\label{eqn:diffalloc}
    \big(p_g^l(t) - f_g^{l'}(y_g^l(t))\big) \text{d}y_{g}^{l}(t) \geq \frac{1}{\alpha_g^l(t)} f_g^{l*'}(p_g^l(t)) \text{d}p_g^l(t)
\end{align}
for all $l\in\mathcal{L}, t=1,\dots,T$. For the energy-procurement operational cost in \eqref{eq:generation cost} and its Fenchel conjugate \eqref{eq:fenchel gen cost} respectively, the following derivatives are:
\begin{align}
    &\label{f_deriv}f_g^{l'}(y_g^l(t)) = 
    \begin{cases}
        0, & y_g^l(t)\in [0, s_l(t))\\
        \pi_l(t), & y_g^l(t)\in [s_l(t), s_l(t)+G_l(t)]
    \end{cases}\\
    &\nonumber\textrm{ and }\\
    &\label{f_deriv_conj}f_g^{l*'}(p_g^l(t)) = 
    \begin{cases}
        s_l(t), & p_g^l(t)\in[0,\pi_l(t))\\
        s_l(t)+G_l(t), & p_g^l(t) \geq \pi_l(t).
    \end{cases}
\end{align}
The derivative of the proposed pricing function \eqref{eqn:newpricing} is:
\begin{align}
\label{eqn:newpricing_deriv}
    &\text{d}p_g^l(y_g^l(t)) = \\*
    &\nonumber\begin{cases}
        \Big( \frac{L_g}{2Rs_l(t)} \Big) \Big( \frac{2R\pi_l(t)}{L_g} \Big)^{ \frac{y_g^l(t)}{s_l(t)}}\\
        \nonumber\times\ln{\Big(\frac{2R\pi_l(t)}{L_g}\Big)}\text{d}y_{g}^{l}(t), \hspace{69pt} y_g^l(t) < s_l(t), \\
        \Big( \frac{L_g-\pi_l(t)}{2R(s_l(t)+G_l(t))} \Big) \Big( \frac{2R(U_g-\pi_l(t))}{L_g-\pi_l(t)} \Big)^{ \frac{y_g^l(t)}{s_l(t)+G_l(t)}}\\
        \nonumber\times\ln{\Big(\frac{2R(U_g-\pi_l(t))}{L_g-\pi_l(t)}\Big)}\text{d}y_{g}^{l}(t), \hspace{48pt}y_g^l(t) \geq s_l(t), \\
        \end{cases}\\
    &\nonumber\textrm{where } R=\sum_{\mathcal{L}}M_l(C_l+E_l+\frac{1}{M_l}).
\end{align}

When $y_g^l(t) < s_l(t)$, $f_g^{l'}(y_g^l(t)) = 0$ and $f_g^{l*'}(p_g^l(t)) = s_l(t)$. As such, after inserting the derivative \eqref{eqn:newpricing_deriv} in \eqref{eqn:diffalloc}, we can show that the Differential Allocation-Payment Relationship holds when $\hat{\alpha}_g^{l}(t)\geq \ln{\Big(\frac{2R\pi_l(t)}{L_g}\Big)}$ as long as $R\geq\lceil\frac{eL_g}{2\max_{\mathcal{L},\mathcal{T}}\pi_l(t)}\rceil$. The constraint on $R$ ensures $\hat{\alpha}_g^{l}(t)\geq1$. 

Similarly, when $y_g^l(t) \geq s_l(t)$, $f_g^{l'}(y_g^l(t)) = \pi_l(t)$ and $f_g^{l*'}(p_g^l(t)) = s_l(t)+G_l(t)$. As such, after inserting the derivative \eqref{eqn:newpricing_deriv} in \eqref{eqn:diffalloc}, we can show that the Differential Allocation-Payment Relationship holds when $\hat{\hat{\alpha}}_g^{l}(t)\geq \ln{\Big(\frac{2R(U_g-\pi_l(t))}{L_g-\pi_l(t)}\Big)}$. Now, let $\alpha_g^{l}(t)=\max\{\hat{\alpha}_g^{l}(t), \hat{\hat{\alpha}}_g^{l}(t)\}$ and because \eqref{eqn:diffalloc} holds for the proposed pricing function, operational cost function, and Fenchel conjugate, the remainder of the proof follows from Lemma 1 and Theorem 2 in \cite{IaaS}.
\end{proof}

\begin{corollary}
If the final demand for energy procurement $y_g^l(t)$ for a given day is less than the available solar $s_l(t)$, the marginal pricing function \eqref{eqn:newpricing} is $\alpha_2$-competitive ($\alpha_2 < \alpha_1$) in social welfare when selling limited resources with the piecewise linear operational cost in \eqref{eq:generation cost} where
\begin{align}
\nonumber
    \alpha_2= 2\max_{\mathcal{L},\mathcal{T}}\Big\{\ln{\Big(\frac{2\sum_{\mathcal{L}}M_l(C_l+E_l+\frac{1}{M_l})(\pi_l(t))}{L_g}\Big)}\Big\}.
\end{align}
\end{corollary}

In the previous proposition, the pricing function \eqref{eqn:newpricing} relies on complete knowledge of the solar generation $s_l(t)$. If the system has inaccurate solar irradiation forecasts, the solar generation could be overestimated and resources are over-allocated resulting in infeasible solutions, which our online solution should avoid at all costs; or solar generation is underestimated and prices are set too high and the system rejects users that should otherwise be accepted. We analyze the case where we have a forecast of the solar generation each day in terms of a confidence interval. We do not assume a specific solar irradiance forecasting method; rather, we make use of a confidence interval for the potential solar each day as yearly solar irradiance recordings can provide minimum and maximum bounds for any given day. Additionally, out method assumes that these confidence regions are tightening as the day progresses. 
In this paper, we assume that the solar forecast for a future time $t$ increases in accuracy as the current time $t_{current}$ approaches $t$. Specifically, the solar forecast takes the following form:
\begin{align}
    s_l(t) \in [\underline{s}_l(t,t_{current}) , \overline{s}_l(t,t_{current})],
\end{align}
for $t=1,\dots,T$ and $1\leq t_{current}\leq t$. Here, $s_l(t)$ is the actual solar generation at time $t$ and the terms $\underline{s}_l(t,t_{current})$ and $\overline{s}_l(t,t_{current})$ are lower and upper bounds given by the forecast, respectively, at an earlier time $t_{current}$. We assume that the forecast is improving, specifically $\overline{s}_l(t,t_{current})$ is non-increasing and $\underline{s}_l(t,t_{current})$ is non-decreasing as $t_{current}$ approaches $t$. 

To account for the dependence of the solar forecast on the current time, the marginal pricing function \eqref{eqn:newpricing} is now written as $p_g^l(y_g^l(t),t_{current})$. To avoid possible infeasible allocations associated with overestimation of solar availability, we analyze the performance of pricing function \eqref{eqn:newpricing} that conservatively uses the underestimate of the solar generation, $\underline{s}_l(t,t_{current})$, in Proposition \ref{piecewise linear price estimate}. Fig. \ref{fig: solar_example} shows how the pricing function \eqref{eqn:newpricing} changes as the solar forecast improves. 

\begin{figure}[h]
    \centering
    \includegraphics[width=0.95\columnwidth]{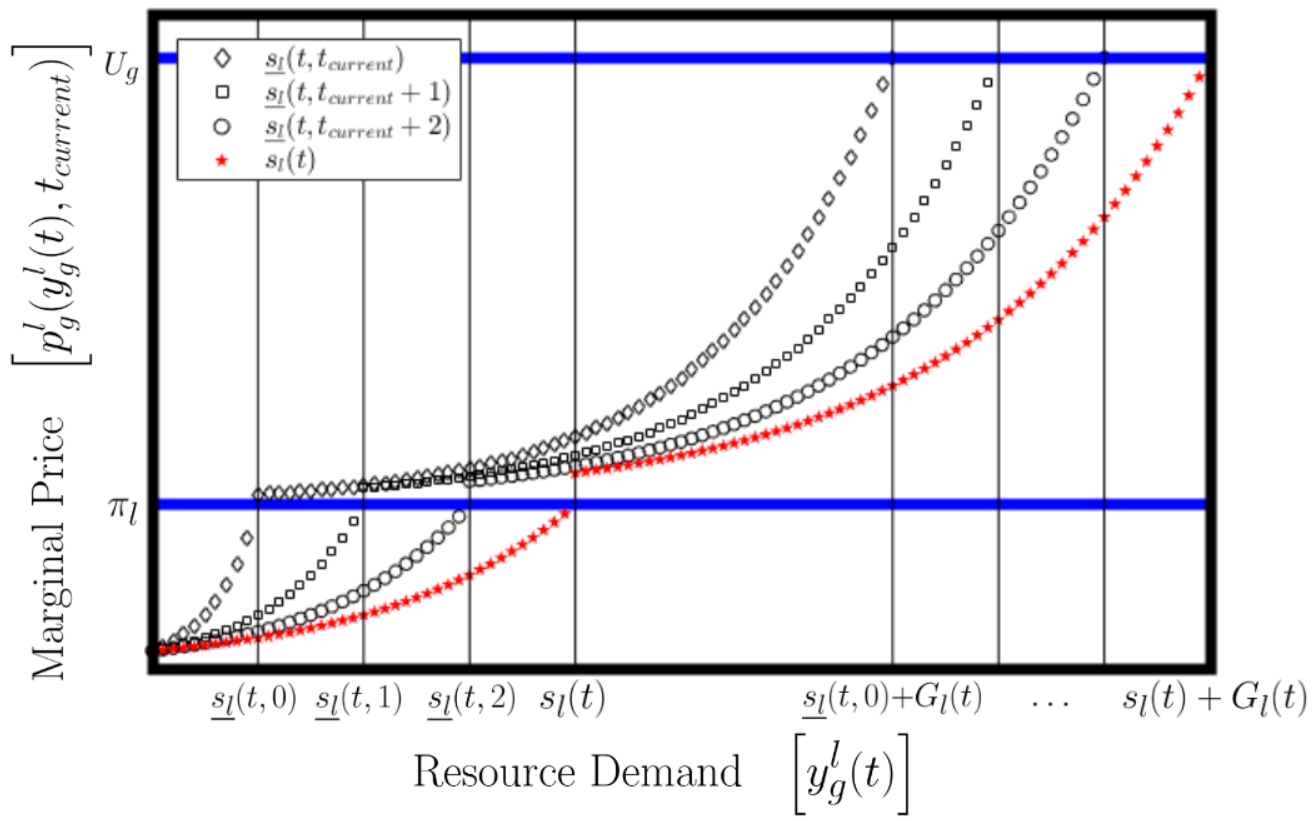}
    \caption{Pricing function with solar forecast $\underline{s}_l(t,t_{current})$.}
    \label{fig: solar_example}
\end{figure}

\begin{proposition}
\label{piecewise linear price estimate}
The marginal pricing function \eqref{eqn:newpricing} with an underestimate of solar generation, $\underline{s}_l(t,t_{current})$, is $\alpha_3= 2\max_{\mathcal{L},\mathcal{T}} \big\{\alpha_g^l(t)\big\}$ competitive in social welfare when selling limited resources with the operational cost in \eqref{eq:generation cost}
where
\begin{align}
\label{eqn:alpha_def}
    &\alpha_g^l(t)= \\
    &\nonumber\begin{cases}
        \max \bigg\{ \frac{\overline{s}_l(t,1)}{\underline{s}_l(t,1)}\ln(\frac{2R\pi_l(t)}{L_g}), \frac{\overline{s}_l(t,1)+G_l(t)}{\underline{s}_l(t,1)+G_l(t)}\ln(\frac{2R(U_g-\pi_l(t))}{L_g-\pi_l(t)}) \bigg\},\\ \hspace{183pt}\underline{s}_l(t,1) \neq 0,\\
        \frac{\overline{s}_l(t,1)+G_l(t)}{G_l(t)}\ln(\frac{2R(U_g-\pi_l(t))}{L_g-\pi_l(t)}), \hspace{60pt} \underline{s}_l(t,1) = 0,
    \end{cases}
\end{align}
with the assumption $R\geq\lceil\frac{eL_g}{2\max_{\mathcal{L},\mathcal{T}}\pi_l(t)}\rceil$.
\end{proposition}

\noindent \textit{Proof.} The proof is omitted for brevity. The proof follows the same structure as Proposition \ref{piecewise linear price} using the lower bound solar generation estimates instead of the actual solar generation amounts as well as the non-increasing and non-decreasing properties of the upper and lower estimates, respectively.

We would like to note the significance of Proposition 1 and Proposition 2 in the following. Namely, our competitive ratio results ensure the social welfare generated by the approximate online solution (that runs in real-time) cannot deviate too far from the social welfare generated by the oracle offline solution. The results in Proposition 1 and Proposition 2 ensure the online system, which acts without knowledge of future arrivals, performs within a constant factor of the offline/oracle system. Furthermore, the competitive ratios are worst case bounds on performance. That is, if this pricing scheme is used in a real scenario, even the social welfare generated with respect to an adversarially chosen arrival sequence is within the constant $\alpha$ of the optimal oracle solution.

\section{Experimental Evaluation}
\label{section:numerical}

\subsection{The Case of Variable Arrival Patterns}
\label{section:simulation CEC}

In this section, we present a comparison of our online pricing mechanism against an online certainty equivalent controller (CEC) for a downtown 
parking facility to show the performance our mechanism under different arrival statistics. CEC is an approximate dynamic programming (ADP) technique that replaces all future uncertain quantities with some typical values, more specifically, the expected values. In this case, we assume that the facility has arrival patterns following the distributions 
in Figure \ref{fig: distributioni}. 
\begin{figure}[h]
    \centering
    \includegraphics[width=0.95\columnwidth]{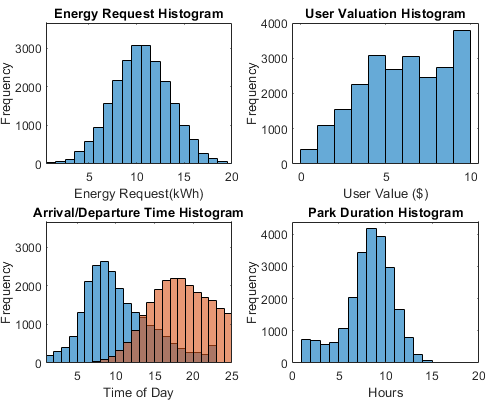}
    \caption{Top Left: Energy requests. Top Right: User valuations. Bottom Left: Arrival times (blue) and departure times (red). Bottom Right: Length of stay.}
    \label{fig: distributioni}
\end{figure}
In this example, there is 1 parking facility with 5 EVSEs and 4 cables per EVSE (i.e., there are 20 parking slots available in each time period). The facility can purchase energy from the Los Angeles grid at a cost of $\$0.127$/kWh. Lastly, the facility has a 32 kW rooftop solar generation system that follows a production curve from an LA location in January 2018 \cite{Solar_PVWATTS}. We assume standard crystalline silicon panels with 14$\%$ system loss due to shading, wiring, connections, mismatch, and degradation. We simulated the CEC and pricing mechanism for 2 different arrival count distributions (as shown in Fig. \ref{fig: CEC vs Pricing}). To demonstrate the value of adverserial solutions like ours in situations when the future is hard to predict, the distribution exhibits a larger variance under the second scenario. 
\begin{figure}[h]
    \centering
    \includegraphics[width=0.9\columnwidth]{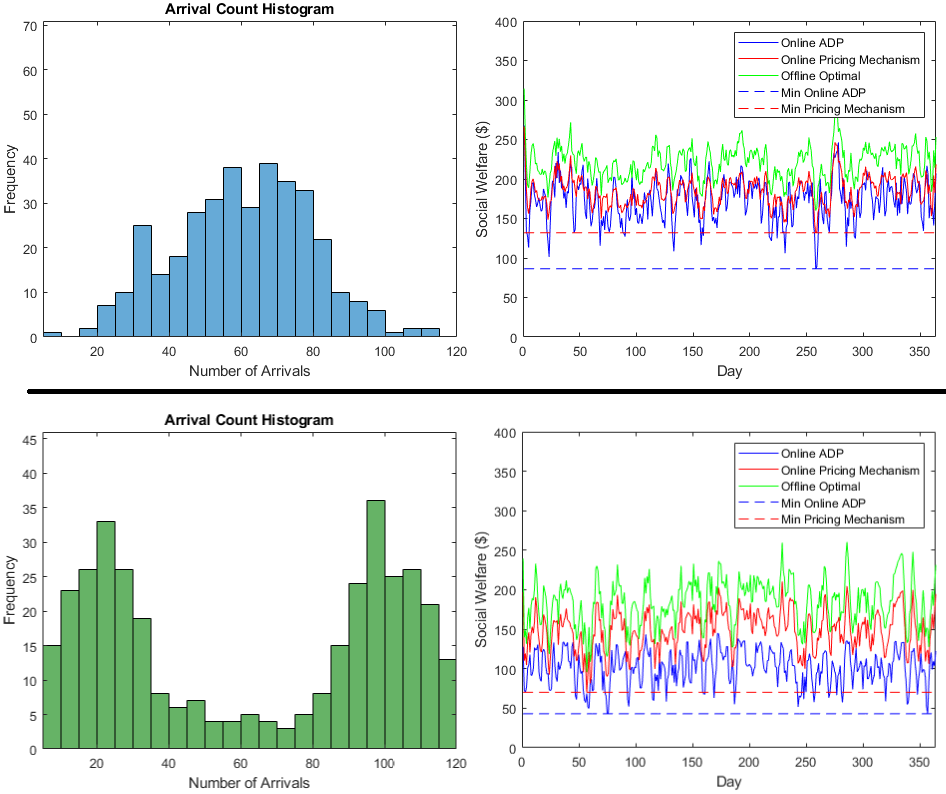}
    \caption{Top: Social welfare for facility with unimodal arrival count distribution. Bottom: Social welfare for facility with bimodal arrival count distribution.}
    \label{fig: CEC vs Pricing}
\end{figure}


From our results in Figure \ref{fig: CEC vs Pricing}, it is evident that in the case of higher variance arrival patterns, our pricing mechanism outperforms a controller that is dependant on expected statistics. If a parking facility does not have consistent arrival statistics each day, our pricing mechanism performs better because it accounts for worst-case arrival patterns. Additionally, in the unimodal and bimodal arrival count cases, the minimum daily social welfare of our mechanism is  larger than that of the CEC ADP. 

\subsection{Comparison with First-Come-First-Serve Strategy}
\label{section:fcfs}
In this section, we present a comparison of our online pricing mechanism against the first-come-first-serve (FCFS) strategy that is commonplace in many EVSE equipped charging facilities. Specifically, we highlight the performance of our mechanism over varying demand levels to show the effectiveness of our mechanism when the infrastructure becomes congested. In the FCFS strategy, an arriving EV selects the closest available parking spot and begins charging immediately (without any controller directing them). In this test case, we assume the arrivals' energy requests, valuations, and durations follow the same statistics as in Fig. \ref{fig: distributioni}; however, in this simulation, we directly control the number of arrivals each day (to highlight different demand levels) and we limit the arrival times to 8:00am-10:00am and limit the departure times to after 10:00am (thus showcasing the performance of FCFS and our mechanism when large quantities of vehicles arrive in a short time period each morning). We assume the parking infrastructure has 15 EVSEs with 4 cables each, yielding 60 parking spots total. 
\begin{figure}[]
    \centering
    \includegraphics[width=0.9\columnwidth]{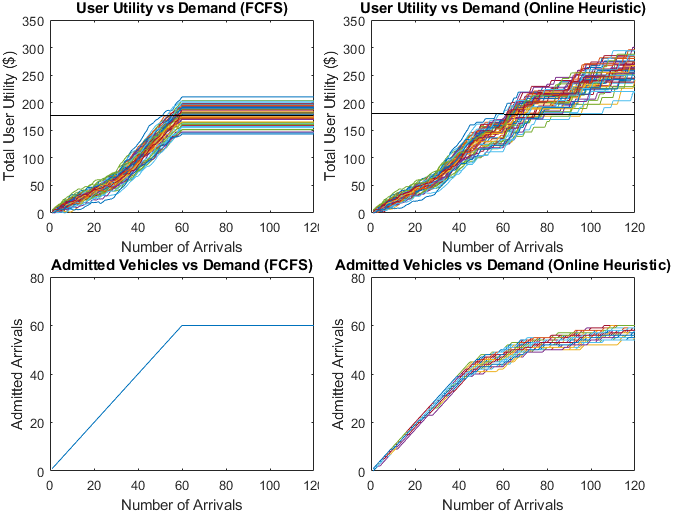}
    \caption{Top Left: FCFS's total user utility vs. number of daily arrivals (50 scenarios simulated). Top Right: Our online heuristic's total user utility vs. number of daily arrivals (50 scenarios simulated). Bottom Left: FCFS's number of admitted EVs vs. demand level (identical result for all 50 scenarios). Bottom Right: Our online heuristic's number of admitted EVs vs. demand level (50 scenarios simulated).}
    \label{fig:fcfs_demand}
\end{figure}

We compare the total user utility yielded from the FCFS strategy to our online heuristic with demand increasing from 1 to 120 arrivals each day. As shown in Fig. \ref{fig:fcfs_demand}, the total user utility increases steadily as the number of EVs entering the system increases. However, when the demand for the EVSEs is high, our online mechanism is able to filter out low value arrivals to admit higher value arrivals instead, and yield higher total utility. It is worth noting that the FCFS strategy yields similar total utility if the demand for charging is low; this is because FCFS admits all arrivals as long as there are open parking spots. 

Additionally, in Fig. \ref{fig:fcfs_solar}, we show the fraction of behind-the-meter solar used by our mechanism and FCFS for a day when there are 100 arriving EVs. From the plot, it is evident that our online mechanism is able to utilize significantly more solar energy, thus eliminating the need to send large amounts of excess energy back to the distribution grid. Specifically, our mechanism is able to schedule charging to time slots when there is available solar, while FCFS is not able to schedule charging times. The results in this section show that our mechanism outperforms the commonplace FCFS strategy in congested facilities in addition to better utilizing behind-the-meter solar.
\begin{figure}[]
    \centering
    \includegraphics[width=0.9\columnwidth]{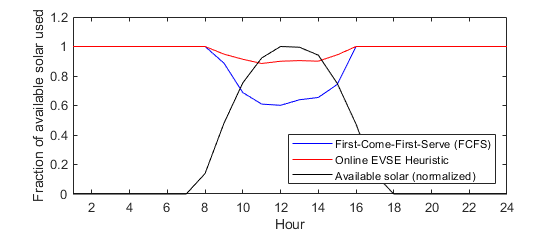}
    \caption{Comparison of behind-the-meter solar energy usage for a day with 100 arrivals.}
    \label{fig:fcfs_solar}
\end{figure}
\subsection{Multi-Facility Test Case}
\label{section:Benefits}
In this section, we present a multi-facility example located in downtown Los Angeles. Specifically, we look at 6 parking facilities with varying rooftop generation amounts. We assume standard crystalline silicon panels with 14$\%$ system loss. Facility 6 has a 75kW solar generation system, facility 5 has 60kW, facility 4 has 45kW, facility 3 has 30kW, facility 2 has 15kW, and facility 1 does not have any solar generation. We examined a 20 day period with 600 arrivals each day. Each arrival has valuation in $[\$1,\$10]$, energy request in [1, 20] kWh, and stay length in [1, 8] hours. Each of the 6 facilities has 8 SOMC EVSEs each equipped with 4 cables. Furthermore, each facility purchases electricity from the Los Angeles grid at $\$$0.127/kWh. We examine the performance of the system with transformer capacity limits of 75kVA.

Figure \ref{fig: 6_loc_100_50} shows the total user utility, social welfare, and electricity cost for 20 days. An observation worth noting is the total user utility from our updated pricing function is always larger than that of a solar agnostic pricing framework. This is due to setting lower prices on the electricity generation resource when there is solar available as seen in Figure \ref{fig: pricing_example} (our previous work was agnostic to the free solar generation). Moreover, over the 20 day period, our updated pricing mechanism admits 387 arrivals on average while our previous work only admits 369 arrivals on average. As such, our updated mechanism is favorable for users of the system as prices are lower and more users are admitted. Additionally, 
our improved mechanism is able to utilize more solar, reducing reliance on the local grid.


\begin{figure}[h]
    \centering
    \includegraphics[width=0.85\columnwidth]{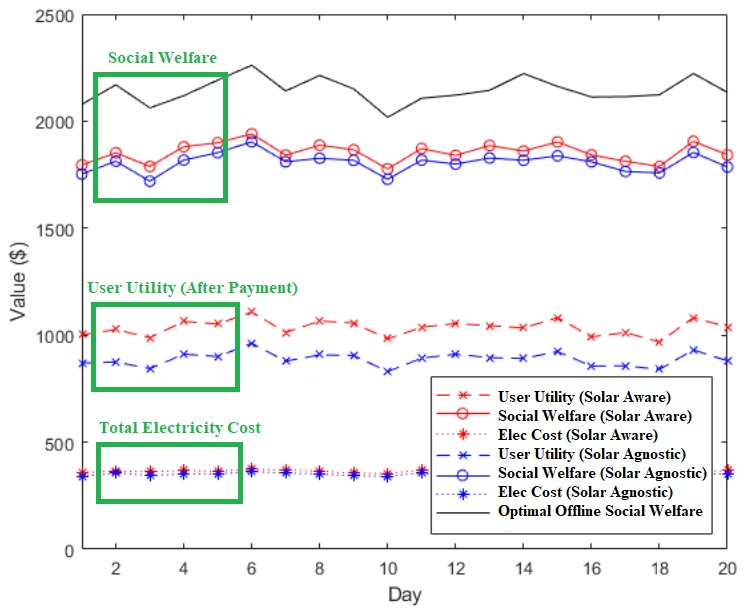}
    \caption{Multi-facility example. Red: solar aware pricing. Blue: solar agnostic pricing.}
    \label{fig: 6_loc_100_50}
\end{figure}



\subsection{Importance of Accurate Departure Times}
\label{section: departure reports}

In this section, we discuss the effect of inaccurate departure time reporting. For example, consider the case when a user reports that she will exit the system by 4:00pm; however, she gets delayed and cannot remove her EV until 5:00pm or later. This affects the reservation system because there might be a reservation for the EVSE at that timeslot. To avoid these reservation collisions due to delayed departures, we analyze the performance of the system with extra hours added to each arrival's stay length as a ``buffer" to prevent double-allocations. We examined the same test case as Section \ref{section:Benefits} with buffer sizes of 1 and 2 hours added to each arrival's stay length. Adding 1 and 2 hour departure buffers yielded average social welfare losses of $16\%$ and $29\%$, respectively.



\subsection{Infrastructure Recommendation}
\label{section:Investment}
In this section, we demonstrate the importance of infrastructure planning in order to maximize the smart charging capabilities of a parking facility. A facility with too few EVSEs will limit the users' utilities as well as the smart charging potential. Conversely, installing too many EVSEs results in idle chargers. As such, we perform a cost-benefit analysis to determine the number of cables at each SOMC EVSE as well as the number of EVSEs that should be installed at a facility in order to maximize social welfare over an extended period. Specifically, we are simulating the same system as described in Sections \ref{section:Benefits} and \ref{section: departure reports}; however, we have increased the duration to 2 years (730 days). Furthermore, we are including initial and recurring costs relating to a parking structure equipped with SOMC EVSEs. These costs include EVSE unit costs, installation costs, electricity consumption, maintenance, and networking costs. In the following, we are looking to choose the constraint variables $C_l$ and $M_l$ (number of cables per EVSE and number of EVSEs, respectively) that maximize users' utilities minus the aforementioned investment and operational costs. We use $I_{C}$ to denote the EVSE unit investment per cable and $I_{M}$ to denote the installation cost per SOMC EVSE. Additionally, $I_{m,n}$ represents a recurring infrastructure maintenance and networking cost per EVSE. As such, the infrastructure investment cost can be written as:
\vspace{0.5pt}
\begin{align}
    \label{eqn:investment obj}
    &\sum_{\mathcal{L}}\big( I_{C}C_l +I_{M} \big)M_l + T\sum_{\mathcal{L}} I_{m,n} M_l.
\end{align}

\noindent In \eqref{eqn:investment obj}, the first term is the initial investment cost for the SOMC EVSE hardware and installation and the second term represents the recurring maintenance and networking cost.
\begin{figure}[]
    \centering
    \includegraphics[width=0.9\columnwidth]{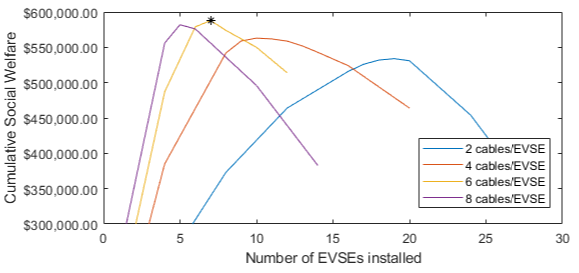}
    \caption{Cost-benefit analysis over 2 years.}
    \label{fig: investment}
\end{figure}
In the following, we assumed each EVSE was a pedestal mounted unit with an installation cost of $\$3,308$ \cite{Investment}. Additionally, each extra cable for each SOMC increased the cost of the EVSE by $\$3,343$ \cite{Investment}. Last, the recurring networking and maintenance fee was assumed to be $\$75$ per month \cite{Investment}. Figure \ref{fig: investment} shows the comparison of total social welfare generated across the entire time period for various levels of infrastructure investments. For this test case, the best result occurred when each location had 7 SOMC EVSEs each equipped with 6 cables. 
With this level of infrastructure, the system did not yield positive social welfare until the second year of operation. As seen in Fig. \ref{fig: investment}, it is clear that sizing a facility for the given use case is critical. Smart charging strategies require a sufficient number of EVSEs to yield maximal benefits; however, welfare decreases if extra EVSEs are purchased and underutilized.

\section{Conclusion}
In this paper, we presented an online pricing mechanism as a solution to the EVSE reservation problem. The online mechanism functions as both an admission controller and a distributor of the facilities' limited charging resources. The work presented in this manuscript complements existing literature in the area and the important characteristics are as follows. First, the mechanism readily accommodates multiple locations, multiple limited resources, operational costs, and variable arrival patterns. The mechanism does not rely on fractional allocations or rounding methods to produce integer allocations in a computationally feasible manner and it never revokes previously made reservations. Moreover, our online mechanism readily handles the inherent stochasticity of the EVSE reservation problem including unknown sojourn times, unknown energy requests, and unknown user valuation functions. The online mechanism can handle adversarially chosen arrival sequences and still generate social welfare within a factor of $\frac{1}{\alpha}$ of the offline optimal. We discussed a competitive ratio as a performance guarantee for the online mechanism compared to the oracle offline solution and provided numerical results showing the efficacy of the mechanism.


\bibliographystyle{IEEEtran}
\bibliography{references}

\end{document}